\documentclass[a4paper,UKenglish]{lipics}
\usepackage{url,graphics,amsmath,amsfonts,amssymb,tikz,verbatim,mdframed,amsthm}
\usepackage[ruled, vlined]{algorithm2e}
\usepackage{cite}
\usepackage{tcolorbox}
\usetikzlibrary{calc}
\usepackage{mathtools}
\usetikzlibrary{patterns}

\newtheorem{claim}{Claim}

\title{Diameter constrained Steiner tree and related problems}
\author{Prashanth Amireddy}
\author{Chetan Sai Digumarthi}
\affil{Department of Computer Science and Engineering,\\Indian Institute of Technology Madras, India.\\
	\texttt{\{cs16b001,cs16b005\}@cse.iitm.ac.in }}

\ArticleNo{1}
\begin{document}
\maketitle

\begin{abstract}
	We give a dynamic programming solution to find the minimum cost of a diameter constrained Steiner tree in case of directed graphs. Then we show a simple reduction from undirected version to the directed version to realize an algorithm of similar complexity i.e, FPT in number of terminal vertices. Other natural variants of constrained Steiner trees are defined by imposing constraints on the min-degree and size of the Steiner tree and some polynomial time reductions among these problems are proven. To the best of our knowledge, these fairly simple reductions are not present in the literature prior to our work.
\end{abstract}
\section{Introduction}
\label{sec:intro}
A \textit{Steiner tree} of a simple undirected graph is a sub-tree that contains a given fixed subset of vertices called \textit{terminals}. The minimum Steiner tree (ST) refers to a Steiner tree with the least possible overall \textit{cost} of its edges. If the edges are unweighted, the cost of each edge is taken as unity, otherwise the cost (or sometimes referred to as weight) would be a positive integer. The ST problem has found numerous applications in a wide variety of fields, like network design and computational biology \cite{HR92, DSR00}. Although this problem is NP-hard \cite{garey1979guide}, several approximation algorithms have been known in the literature \cite{testapprox, imporveapprox, arora}. Most often, all the terminals can be assumed to be leaves without loss of generality \cite{test, testapprox, MARTINEZ2007133}. The solution to the problem may be easy depending on the \textit{topology} of the objective tree as well \cite{ding2011diameter, Melzak1961OnTP}. We study this problem under different constrained settings and design FPT algorithms for it and observe that some of the constrained versions reduce to one another. We begin by defining four modifications of the ST problem we are interested in solving.

\subsection{Constrained ST Problems}

\noindent\fbox{
	\parbox{13.5cm}{
		{\sc Diameter Constrained Steiner Tree (DCST)}  \\
		\textsf{\bfseries Input $(V, E, T, D)$ : } A non-negative-weighted simple undirected graph $G=(V, E)$, a set $T \subseteq V$ of terminals, and an integer $D > 0$. \\
		\textsf{\bfseries Output:} Minimum weight of a sub-tree of $G$ containing all the vertices in $T$ such that its diameter is at most $D$. If no such sub-tree exists, the output is FAIL. 
		
		\textit{Notation -} Here, the weight of a sub-tree $H$ of $G$, denoted by $wt(H)$ refers to the sum of weights of its edges. We will refer to a constrained sub-tree that attains the minimum weight as an optimal Steiner tree.

		\vspace{3mm}
		
		{\sc Directed Diameter Constrained Steiner Tree (DDCST)}  \\
		\textsf{\bfseries Input $(V, E, T, r, D)$:} A non-negative-weightedweighted directed graph  $G=(V,E)$ such that $E(i,i)=0$, and $E(i,j)$ and $E(j,i)$ can be unrelated for $i \ne j$, a set $T \subseteq V$ of terminals, a root $r \in V\setminus T$, and an integer $D > 0$. \\
		\textsf{\bfseries Output:} Minimum weight of a rooted sub-tree (rooted at $r$) of $G$ containing all the vertices in $T$ such that its diameter is atmost $D$. If no such sub-tree exists, the output is FAIL. 
		
		\vspace{3mm}
		
		{\sc Minimum-degree Constrained Steiner Tree (MCST)}  \\
		\textsf{\bfseries Input $(V, E,T,\Delta)$ : } A non-negative-weightedweighted simple undirected graph $G=(V,E)$, a set $T \subseteq V$ of terminals, and an integer $\Delta>0$. \\
		\textsf{\bfseries Output:} Minimum weight of a sub-tree of $G$ containing all the vertices in $T$ such that each internal node in the tree has degree atleast $\Delta$. If no such sub-tree exists, the output is FAIL.

		\vspace{3mm}
		
		{\sc Size Constrained Steiner Tree (SCST)}  \\
		\textsf{\bfseries Input $(V, E,T,\zeta)$ : } A non-negative-weightedweighted simple undirected graph $G=(V,E)$, a set $T \subseteq V$ of terminals, and an integer $\zeta\ge |T|$. \\
		\textsf{\bfseries Output:} Minimum weight of a sub-tree of $G$ containing all the vertices in $T$ such that its size (i.e, no. of vertices in the tree) is atmost $\zeta$. If no such sub-tree exists, the output is FAIL. 
	
	}

}



\subsection{Definition - Diameter of a tree}
In Graph theory, the \textit{diameter} of a weighted directed rooted tree is defined as the longest path (without considering the weights of edges in the path) from root to any leaf. Similarly, the \textit{diameter} of  a weighted undirected tree is defined as the longest path (without considering the weights of edges in the path) between any two leaves. Here, length of a path refers to the number of edges in that path.

\subsection{Instance relaxation}
Suppose that an input to the DCST problem (or to any of the aforementioned problems) has some edge weights in the graph as $\infty$, meaning these edges are not present in the input graph. We will show that such weights can be replaced with some large enough quantity $M$ that is only polynomial in the input size and the graph would still retain the original Steiner tree as the optimal one. 

Although the following claim is proved only for the DCST problem, the same argument works for the other three problems as well. This means that without loss of generality, we can always assume that the input graph is a complete graph.

\begin{claim}
	\label{claim:1}
	Denoting the input size of a problem by $m$, if the DCST problem where no edge weights are $\infty$ can be solved in time $p(m)$, then the general DCST problem can also be solved in time $O(p(2.m^2))$. The converse is also true. 
\end{claim}

\begin{proof}
	The proof of the converse is trivial. For the forward direction, assume that $\cal A$ is an algorithm that solves the DCST problem with finite weights in time $p(.)$ in input size. We claim that Algorithm \ref{algo:1} is the desired algorithm that solves DCST.
	
	\begin{algorithm}
		\caption{A reduction from DCST to $\cal A$}
		\textbf{Input:} $(V, E, T, D)$\\
		 $M \gets 1+\sum\limits_{{i,j}: E(i,j)<\infty} E(i,j)$\\
		 $V' \gets V$\\
		 $E' \gets E$\\
		 \If{$E(i,j)=\infty$}{$E'(i,j) \gets M$}
		 $X \gets {\cal A}(V', E', T, D)$\\
		 \If{$X < M$}{\Return $X$}\Else{\Return FAIL}
		 \label{algo:1}
	\end{algorithm}
	
	To observe the correctness of the above algorithm, we consider two cases. First, suppose that $X \ge M$. Then clearly the above algorithm FAILs. For contradiction, suppose that the algorithm errs on this input. It means that there is a Steiner tree $H$ of $G$ with $wt(H)$, which trivially would be strictly less than $M$. By the construction of $G'$, the graph $H$ is also a sub-tree of $G'$. Therefore, the optimal Steiner tree of $G'$ has to be at most $X=wt(H) < M$, a contradiction. Hence, we can see that the algorithm behaves correctly in the case $X \ge M$. Now, consider the case when it happens that $X < M$. Denoting the (constrained) optimal Steiner tree of $G'$ by $H'$, we have $X = wt(H') < M$, meaning that all the edges in $H'$ are strictly less than $M$. This implies that $H'$ is a sub-tree of $G$ as well, suggesting that the optimal weight of a Steiner tree for $G$ would have to be at most $X = wt(H')$. But notice that it cannot be strictly less than $X$. Therefore, it must be that the optimal weight of a Steiner tree for the problem $(V, E, T, D)$ is equal to $X$, which is exactly the output of the above algorithm.
	
	A bulk of the running time of the above algorithm comes from the sub-routine $\cal A$. The input size to this sub-routine is at most $ \log M \le 2.m$ times the original input size $m$, hence the overall run time is $O(p(2.m^2))$.
	
\end{proof}





\subsection{Hardness of the problems}
Since the standard (unconstrained) Steiner tree problem can be seen as a version of DCST with $D=\infty$ or a version of MCST with $\Delta = 1$, or a version of SCST with $\zeta = \infty$, all the four constrained problems we defined are at least as hard as the general Steiner tree problem.
 As it is known that the general problem is NP-hard and is $W[2]$-hard when parameterized by the number of non-terminals (see the books \cite{cygan2015parameterized, Downey_Fellows_1999}), so are all the above problems. In the next section, we give a dynamic programming algorithm for DDCST of running time $3^{|T|}.n^{O(1)}$, which is of similar complexity as for the general problem \cite{cygan2015parameterized}. This puts DDCST in FPT when parameterized by the number of terminals.

\section{A dynamic programming algorithm for DDCST}
\label{sec:2}

The following algorithm is heavily inspired by similar algorithms for the general Steiner Tree problem \cite{cygan2015parameterized} and a network optimization problem \cite{fontes2007diameter}.

Let $(V, E, T, r, D)$ be a DDCST instance. For any subset $S\subseteq T$, a vertex $u\in V$, and an integer $1\le d\le D$, define $f(S,u,d)$ to be the minimum weight of a directed sub-tree of $(V, E)$ rooted at $u$ and  and all the vertices in $S$ are reachable from $u$ in at most $d$ edges of the tree. We may assume such a tree always exists by the instance relaxation idea (Claim \ref{claim:1}). Immediately note that our goal is to compute the value of $f(T,r,D)$. The smaller problems for the DP are the instances with smaller $|S|$ or $d$ values. 

 Recall that the weight of an edge $(u, v)$ is denoted by $E(u,v)$. Accordingly, the base cases for the function $f$ are defined as $f(S,u,1)=\sum \limits_{v \in S} E(u,v)$ for all subsets $S \subseteq T$. To give some intuition about the DP, suppose we are given all optimal Steiner trees corresponding to all $S'\subseteq S$ of diameter at most $d$ and we want to construct an optimal Steiner tree for $S$ of diameter at most $d$. The idea is to consider two cases -- based on whether the degree of the root in the tree is exactly equal to 1. More formally, we will prove the following lemma.

 \begin{figure}[t!]
 	\centering
 	\begin{subfigure}[b]{0.5\textwidth} 
 		\includegraphics[height=2.0in]{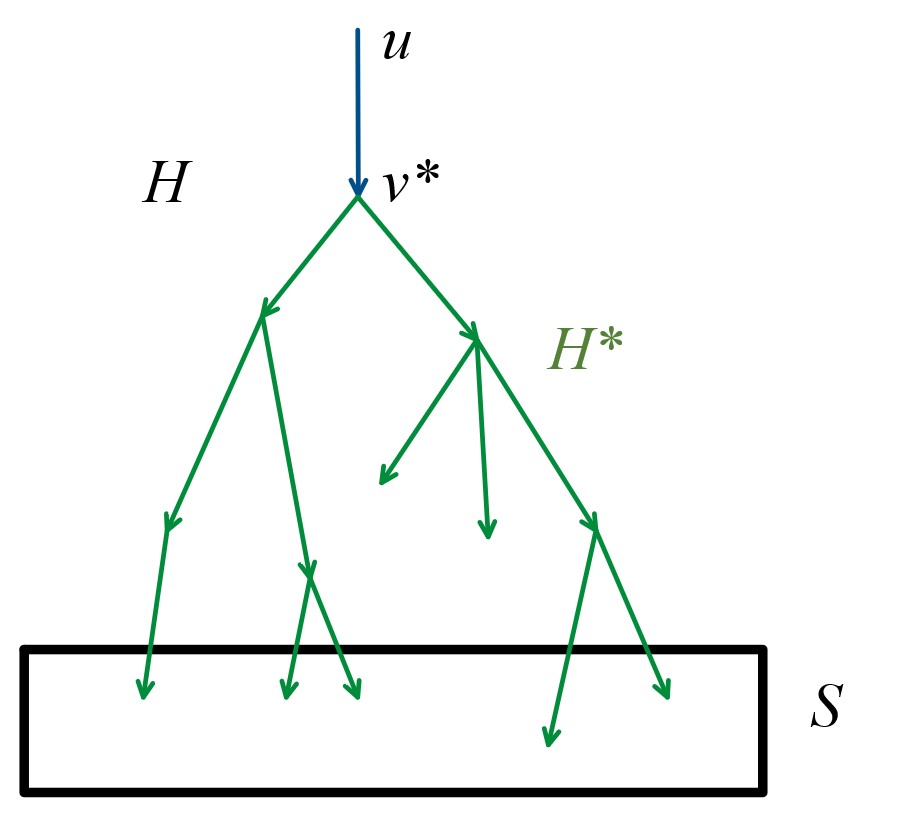}
 		\caption{ }
 		\label{fig:1a}
 	\end{subfigure}%
 	~
 	\begin{subfigure}[b]{0.5\textwidth}
 		\includegraphics[height=1.9in]{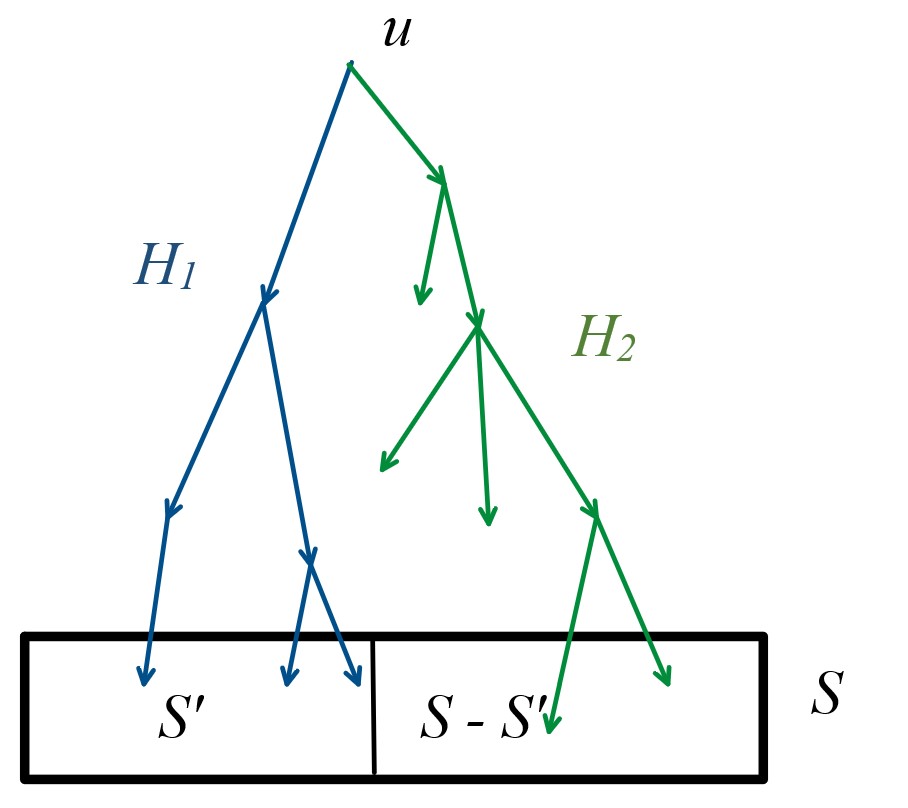}
 		\caption{ }
 		\label{fig:1b}
 	\end{subfigure}
 	
 	\begin{subfigure}[b]{0.5\textwidth}
 		\includegraphics[height=1.7in]{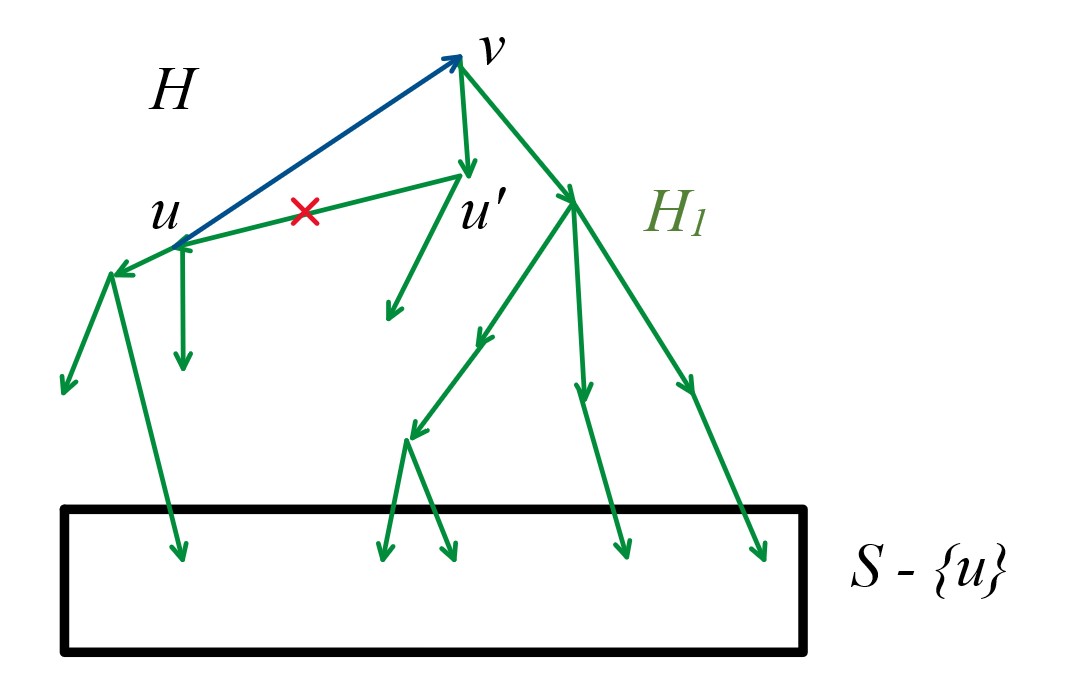}
 		\caption{ }
 		\label{fig:1c}
 	\end{subfigure}%
 	~ 
 	\begin{subfigure}[b]{0.5\textwidth}
 		\includegraphics[height=1.8in]{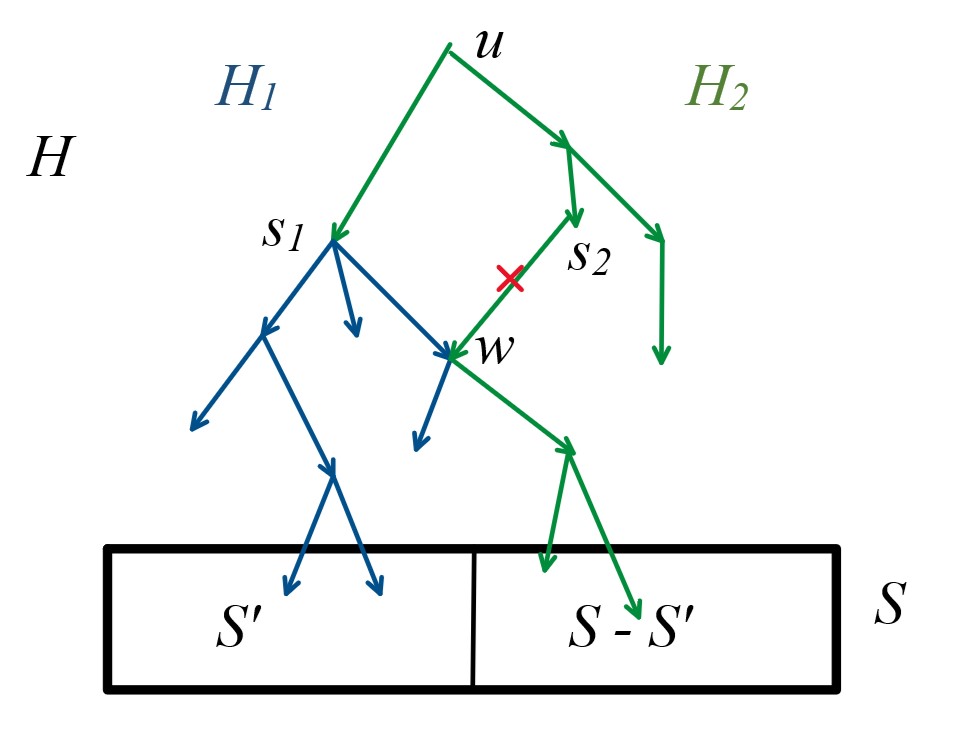}
 		\caption{ }
 		\label{fig:1d}
 	\end{subfigure}
 	\caption{Sub-problems of DDCST}
 \end{figure}
 
 \begin{lemma}
 	\label{lem:1}
 	For any $d\ge 2$ and non-empty subset $S\subseteq T$, 
 	\[f(S,u,d)=\min \bigg \{\min_{\substack{ v \in V}} \{  E(u, v) + f(S\setminus \{u\}, v, d-1)  \}, \min_{\substack{\Phi \neq S' \subsetneq S }} \{ f(S',u,d)+f(S\setminus S',u,d)  \} \bigg \}\]
 \end{lemma}

\begin{proof}
	Note that a non-empty partition of $S$ is possible only when $|S| \ge 2$. When $|S|=1$, we only consider the minimum over the first part of the RHS. To prove this lemma, we will argue that LHS = RHS.
	
	\vspace{2mm}
		
	\noindent \textbf{LHS $\ge$ RHS}: Suppose that the optimal Steiner tree corresponding to the LHS is $H$, i.e, $wt(H) = f(S, u, d)$. We consider two cases based on whether the degree of $u$ in $H$ is exactly one or not. First, if the degree of the root $u$ in $H$ is exactly 1(see Figure \ref{fig:1a}), say its only neighbor is $v^*$. In this case, consider the sub-tree $H^*$ of $H$ by removing the vertex $u$ and the incident edge from it. By the definition of $H$, we have that $H^*$ is rooted at $v^*$ and all the vertices in $S\setminus \{u\}$ are at a distance of at most $d-1$ from this root. Therefore, $wt(H^*) \ge f(S\setminus \{u\}, v^*, d-1)$ due to the definition of the function $f$. Combining these observations, we get LHS $ = f(S, u, d) = wt(H) = E(u, v^*) + wt(H^*) \ge E(u, v^*) + f(S\setminus \{u\}, v^*, d-1) \ge $ RHS. 
	
	In the second case, when the degree of the root $u$ in $H$ is at least 2 (see Figure \ref{fig:1b}), arbitrarily partition the subset $S$ into non-empty subsets $S'^*$ and $S \setminus S'^*$ such that the least common ancestor of any vertex in $S'*$ and any vertex in $S \setminus S'^*$ is  always $u$. This is possible since the degree of the root is at least 2. Let $H_1$ and $H_2$ refer to the sub-trees of $H$ over the ancestors of $S'^*$ and $S \setminus S'^*$ respectively. In simple terms, $H_1$ and $H_2$ is a \textit{partition} of $H$ and are \textit{connected} only at the root $u$. Clearly, both $H_1$ and $H_2$ have the same root and diameter upper bound as $H$, that is $u$ and $d$ respectively. By the minimality criterion in the definition of $f$, we therefore get $f(S'^*, u, d) \le wt(H_1)$ and $f(S \setminus S'^*, u, d) \le wt(H_2)$. Adding up these inequalities, we get the desired expression: LHS $ = wt(H) = wt(H_1) + wt(H_2)\ge f(S'^*, u, d)+f(S \setminus S'^*, u, d) \ge $ RHS.

	\vspace{2mm}
	
	\noindent \textbf{LHS $\le$ RHS}: We need to show that the LHS is less than or equal to $E(u,v)+f(S \setminus \{u\},v,d-1)$ and $f(S',u,d) + f(S\setminus S',u,d)$ whatever the choice of $v$ or $S'$ may be. 
 	For the first part, we need to show that the LHS is at most $E(u,v)+f(S \setminus \{u\},v,d-1)$. Let the optimum of $f(S \setminus \{u\}, v, d-1)$ be attained via a sub-tree $H_1$ i.e, $wt(H_1) = f(S \setminus \{u\}, v, d-1)$. Now, consider a sub-graph $H$ of $G$ defined as follows(see Figure \ref{fig:1c}): Add the vertex $u$ to $H_1$ if it is not present already and also add the edge $(u, v)$. It is easy to see that all the vertices in $S \setminus \{u\}$ are reachable from $u$ via at most $d-1 +1 = d$ edges in $H$. Note that it can be the case that $H$ is not a tree. But, we can carefully remove some edges (in fact at most one edge) without altering the diameter and the set of terminals reached, other than $u$ itself (if it is a terminal, that is) to make it a sub-tree. To observe this, say $u$ appears in the sub-tree $H$ as a child to some vertex $u'$ (If it does not appear at all, then there is nothing to modify). Remove the edge $(u',u)$ from $H$. This would clearly make it a sub-tree rooted at $u$. Hence, we have the desired inequality $f(S, u, d) \le wt(H) \le E(u, v) + wt(H_1) = E(u, v) + f(S\setminus \{u\}, v, d-1)$. The first inequality comes from the minimality criterion used in defining $f(.)$.
 	
 	For the second part, say $H_1$ and $H_2$ respectively are the optimal Steiner trees corresponding to $f(S',u,d)$ and $f(S\setminus S',u,d)$ (see Figure \ref{fig:1d}). Then, define a sub-graph $H$ of $G$ as the union of $H_1$ and $H_2$. Thus, any vertex in $S$ can be reached from $u$ using at most $d$ edges in $H$. The only problem is that $H$ might not be a tree. To fix this, we can make it a tree without disturbing the diameter or the terminals reached by carefully removing some edges. Suppose there exists two path $u \rightarrow P_1 \rightarrow s_1 \rightarrow w$ and $u\rightarrow P_2 \rightarrow s_2 \rightarrow w$ from the root $u$ to some vertex $w$ in the trees $H_1$ and $H_2$ respectively. Without loss of generality, we can further assume that the intermediate vertices in the two paths are disjoint and that the length of the path $P_1$ is not greater than that of $P_2$. If no such $w$ exists, then clearly, the graph $H$ would already be a tree. We claim that we can safely remove the edge $(s_2,w)$ from $H$. The only vertices that this removal can negatively effect are the terminals in the sub-tree of $H_2$ rooted at $w$. But there would still be paths from $u$ to those vertices of shorter or same length because of the existence of the path $u \rightarrow P_1 \rightarrow w$. By using this principle, we can keep removing edges until the final graph is a tree. We would still have maintained the invariant that all the vertices in $S$ are reachable from $u$ by paths in $H$ of length at most $d$. Finally, $wt(H) \le wt(H_1) + wt(H_2) = f(S', u, d) + f(S \setminus S', u, d)$. Hence, $f(S, u, d) \le wt(H) \le f(S', u, d) + f(S \setminus S', u, d)$.
 	
\end{proof}

The above lemma when translated into a table filling algorithm gives rise to the following Theorem.

\begin{theorem}
	\label{thm:2}
	Given a DDCST instance $(V, E, T, r, D)$, it can be solved in time $O(3^{|T|}.  n^{3})$, where $n=|V|$.
\end{theorem}

\begin{proof}
	The dimensions of the DP \textit{table} are $2^{|T|} \times n \times D$. The base cases $f(S, u, 1)$ can be calculated in linear time. Given all the quantities on the RHS of Lemma \ref{lem:1}, to get the LHS, we would need to take time $2^{|S|} + n$. Hence, the total time taken by the algorithm is at most
	\[ ~\sum_{S\subseteq T}~\sum_{v\in V}  ~\sum_{d=1}^{D} ~ 2^{|S|}.n = D.n^2.\sum_{i=1}^{|T|} {n \choose i}.2^i=D.n^{2}.3^{|T|}=3^{|T|}.n^{3}.\]
	
	In the above analysis, we have assumed that $D \le n$ as otherwise, since if $D$ is larger than $n$, we can always work with $D = n$.
\end{proof}

\section{A reduction from DCST to DDCST}
\label{sec:3}

\begin{figure}[t!]
	\centering
	\begin{subfigure}[b]{1\textwidth}
		\includegraphics[height=2.6in]{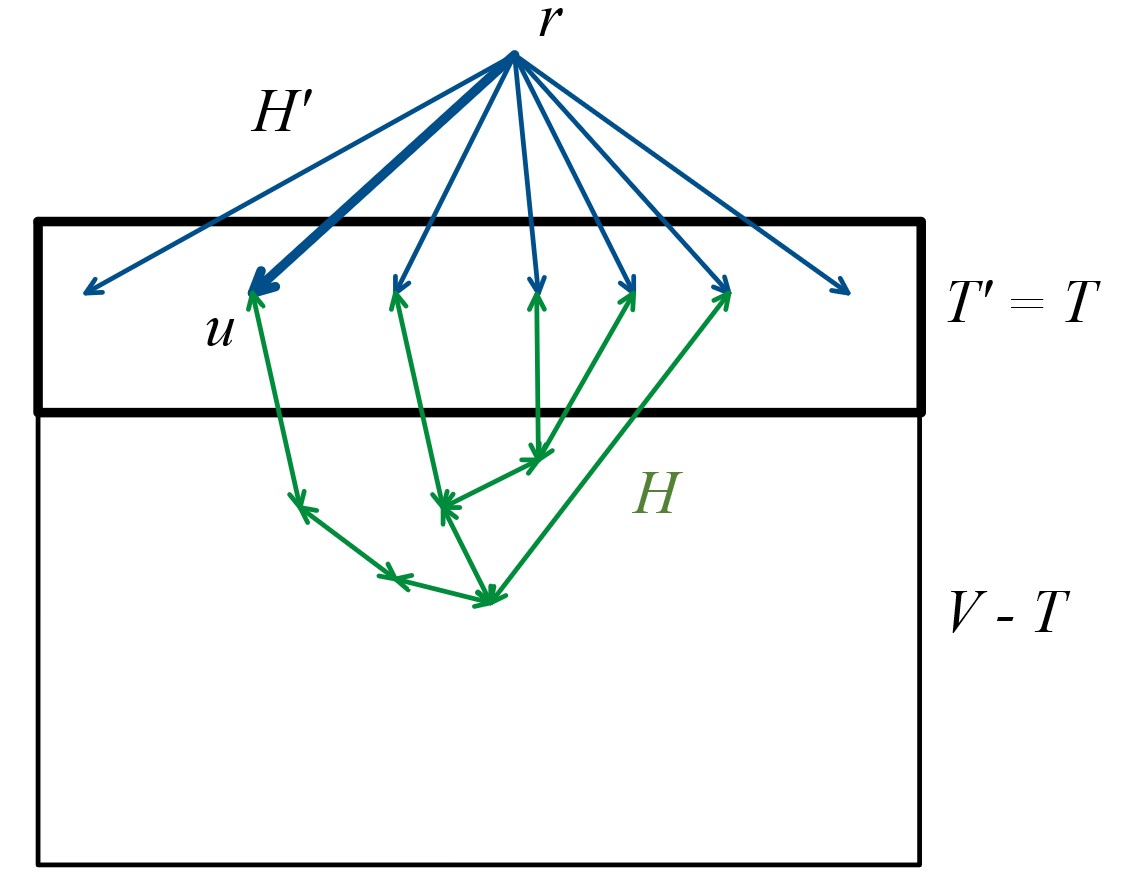}
	\end{subfigure}
	\caption{Undirected to directed DCST reduction}
	\label{fig:2}
\end{figure} 

We note that the above algorithm does not directly work for DCST problem (the undirected version). To show that DCST is also solvable in $3^{|T|} n^{3}$ time, it suffices to prove that the undirected version \textit{reduces} to the directed version, which we describe below.

\begin{theorem}
	\label{thm:3}
	The DCST problem polynomial time Turing reduces to the DDCST problem.
\end{theorem}

\begin{proof}
	Suppose that an algorithm $\cal A$ solves the DDCST problem. Then we claim that the following algorithm (Algorithm \ref{algo:2}) solves the DCST problem.
	\begin{algorithm}
		\caption{A reduction from DCST to DDCST i.e, $\cal A$}
		\textbf{Input:} $(V, E, T, D)$\\
		$M \gets \max \{E(i,j)\}$\\
		$V' \gets V \cup \{r\}$, for a new vertex $r$.\\
		\ForEach{$i,j \in V'$}{\If{$i=j$}{$E'(i,j) \gets 0$} \Else {$E'(i,j) \gets \infty$}}
		\ForEach{$i,j \in V$}{$E'(i,j) \gets E(i, j)$}
		\ForEach{$i \in T$}{$E'(r,i) \gets M.|T|$}
		$X \gets {\cal A}(V', E', T, r, D + 1)$\\
		\Return $X - M.|T|$
		\label{algo:2}
	\end{algorithm}
	
	Before we prove the correctness of this algorithm, note that $M$ is well-defined as we can assume that all the weights are finite without loss of generality (Claim \ref{claim:1}). The correctness boils down to showing that there exists a (directed) Steiner tree of weight at most $X$ and diameter at most $D + 1$ for $(V', E')$ if and only if there exists a (undirected) Steiner tree of weight at most $X - M.|T|$ and diameter at most $D$ for $(V,E)$. For the `if' part of this claim, suppose $H$ was an optimal Steiner tree for $(V, E, T, D)$ with $wt(H) = X - M.|T|$ (see Figure \ref{fig:2}). Consider the directed rooted tree $H'$ that is rooted at $r$ obtained by adding the edge $(r,u)$ to $|H|$ for an arbitrary $u \in T$. The directions for the edges of $H'$ are naturally dictated once a $u$ is fixed. Clearly, $H'$ is a valid Steiner tree for the directed instance of the problem, with diameter at most 1 more than the longest path in the undirected tree $H$ i.e, $D + 1$, and $wt(H') = E'(r,u) + wt(H) = M.|T| + X - M.|T| = X$. 
	 
	 Now, for the other direction of our claim, consider an arbitrary terminal $t \in T$ and consider the Steiner tree made up by the edges $(r, t)$ and $(t,i)$ for all $i \ne t \in T$. It is of diameter 2 $\le D+1$ and of weight at most $M.|T| + (|T|-1).M < 2.M.|T|$, meaning that the optimal Steiner tree for the instance $(V',E',T,r,D+1)$ should be of weight strictly less than $2.M.|T|$. Denoting that optimal Steiner tree by $H'$, we will now argue that the degree of $r$ in $H'$ is exactly 1. For contradiction, suppose this is not the case. That is, $r$ has at least two neighbors in $H'$. But recall that by the construction of $E'$, non-terminals cannot be neighbors of $r$. Hence, $wt(H')$ is at least twice the weight of each root to terminal edge introduced, i.e, at least $2.M.|T|$, thereby contradicting the upper bound we have already established. Hence, the degree of $r$ in $H'$ is exactly 1. Consider the undirected version of the tree obtained by removing the vertex $r$ and the incident edge from $H'$, call it $H$. Clearly, the diameter of $H$ is at most $D + 1 - 1 =D$ and all the terminals are still covered by this tree. Hence, it is a valid solution to the DCST problem for the instance $(V, E, T, D)$ with weight at most $wt(H') - M.|T| \le X - M.|T|$.
	 
	 Finally, we note that this is a polynomial time transformation since only one new vertex and its edges are added to construct the new graph and once the solution of directed instance is known, it involves only a simple subtraction.
\end{proof} 






The above Theorem when combined with Theorem \ref{thm:2} gives the following Corollary.

\begin{corollary}
	\label{cor:4}
	Given a DCST instance $(V, E, T, D)$, it can be solved in time $O(3^{|T|}.  n^{3})$, where $n=|V|$.
\end{corollary}

\section{A reduction from SCST to MCST}
\label{sec:4}
In this section, we show a similar kind of a reduction between the other two constrained Steiner tree problems. To do that, we need the following elementary result from graph theory.

\begin{lemma}
	\label{lem:5}
	Let $T$ be an undirected tree with $p$ internal nodes and $q$ leaves. If all the internal nodes have degree at least $\Delta$, then $q \ge (\Delta -2).p$.
\end{lemma}

\begin{proof}
	Root the tree at an arbitrary internal vertex $u$ and call this tree $T'$. Then, the number of internal nodes and leaves remains the same in $T'$ as $T$ and the (rooted) degree of all the internal vertices is at least $d := \Delta -1$. We will now obtain a bound of $q \ge (d-1).p$ using induction on $p$. The base case $p=1$ is clearly true. Consider an internal node $v$ of $T'$, whose children are all leaves. Let $T^*$ be the tree obtained by deleting all the $x\ge d$ children of $v$ from $T'$. The number of internal nodes and leaves in $T'$ are $p-1$ and $q-x + 1$ respectively. By the induction hypothesis, we have $q-x +1 \ge (d-1).(p-1)$, which gives the desired bound $q \ge (d-1).p$ on expanding.
\end{proof}


\begin{figure}[t!]
	\centering
	\begin{subfigure}[b]{1\textwidth}
		\includegraphics[height=2.4in]{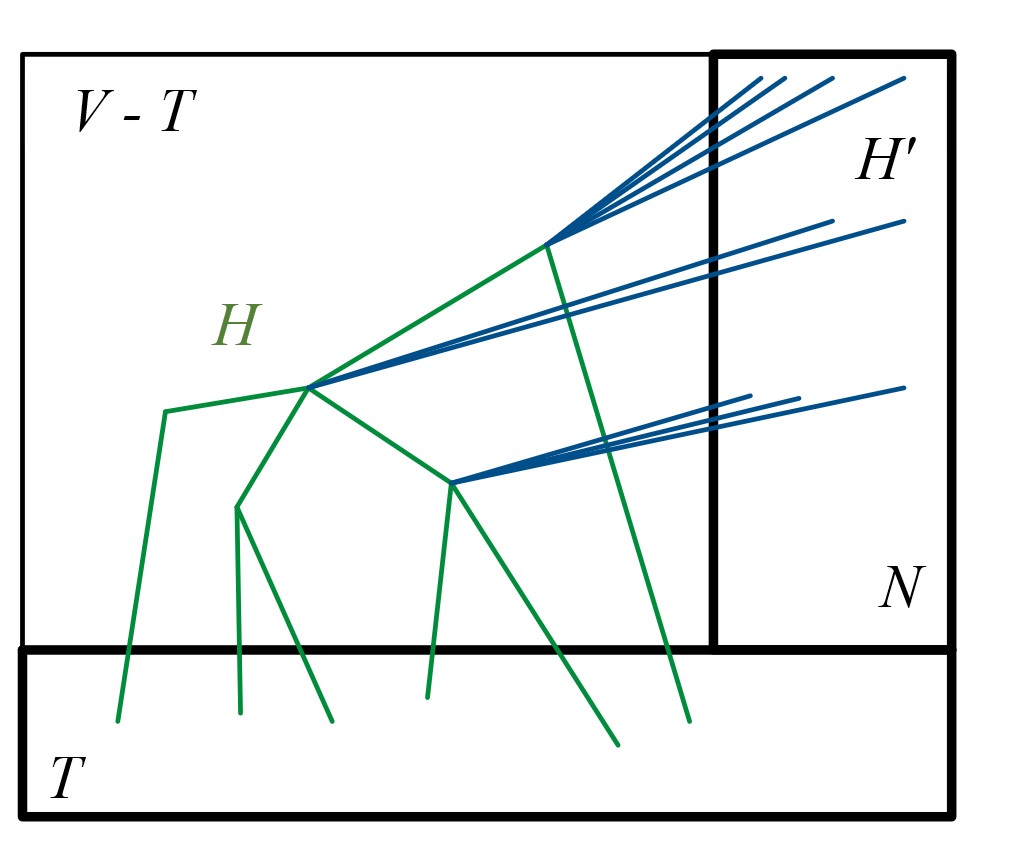}
	\end{subfigure}
	\caption{SCST to MCST reduction}
	\label{fig:3}
\end{figure}

\begin{theorem}
	\label{thm:6}
	The SCST problem polynomial times Turing reduces to the MCST problem.
\end{theorem}

\begin{proof}
		Suppose that an algorithm $\cal A$ solves the MCST problem. Then we claim that the following algorithm (Algorithm \ref{algo:3}) solves the SCST problem.
	\begin{algorithm}
		\caption{A reduction from SCST to MCST i.e, $\cal A$}
		\textbf{Input:} $(V, E, T, \zeta)$\\
		$M \gets \max \{E(i,j)\}$\\
		$\alpha \gets \zeta - |T|$\\
		$\beta \gets 2.|\zeta|$\\
		$N \gets \{\eta_i^j : 1\le i \le \alpha \text{~and~} 1\le j\le \beta\}$, the set of new vertices.\\
		$V' \gets V \cup N$\\
		$T' \gets T \cup N$\\
		\ForEach{$i,j \in V'$}{\If{$i=j$}{$E'(i,j) \gets 0$} \Else {$E'(i,j) \gets \infty$}}
		\ForEach{$i,j \in V$}{$E'(i,j) \gets E(i, j)$}
		\ForEach{$i \in V \setminus T$ and $\eta \in N$}{$E'(i,\eta) \gets 0$\\$E'(\eta, i) \gets 0$}
		$X \gets {\cal A}(V', E', T', 2.\zeta)$\\
		\Return $X$
		\label{algo:3}
	\end{algorithm}

	By instance relaxation (Claim \ref{claim:1}), we are free to give edge weights as 0 and $\infty$, we can appropriately scale up or scale down the weights if needed. We will make use of a trick that was alluded to in the Introduction -- we will assume that in optimal Steiner trees (be it SCST or MCST), the leaves are exactly the set of terminals. This trick although generally invoked for the standard Steiner tree also naturally extends to these constrained versions. Now, to show the correctness of the above algorithm, we will argue that there exists a Steiner tree of size at most $\zeta$ in $G$ and weight at most $X$ if and only if there exists a Steiner tree of the same weight with minimum-degree at least $2.\zeta$ and weight at most $X$ in $G'$.
	
	To prove the forward implication, let $H$ be a Steiner tree of size at most $\zeta$ in $G$ with $wt(H) \le X$ (see Figure \ref{fig:3}). Since, the leaves of $H$ should exactly correspond to the terminals, we can say that the number of internal nodes of $H$, say $I$ is at most $\zeta - |T|$. We now describe a Steiner tree $H'$ for the MCST instance $(V', E', T', 2.\zeta)$. Besides all the vertices and edges that are in $H$, we add additional edges to $H'$. From each internal node $u$ of $H$, add edges to $2.\zeta$ many new vertices $\eta^i_j$'s. While doing this, we can make sure that we do not pick the same $\eta^i_j$ twice since the overall number of new vertices available is equal to $(\zeta - |T|).2.\zeta$ which greater than or equal to $I \times 2.\zeta$ since we have $I \le \zeta - |T|$. If there are some unused $\eta^i_j$'s left over, we can make them adjacent to an arbitrary internal node of $H$. This finishes the construction of $H'$. It can be observed that the weight of $H'$ is the same as that of $H$ since all the new edges we have added are of weight zero. Moreover, all the terminals $T' = T \cup N$ are contained in $H'$, and the degrees of all the internal vertices are at least $2.\zeta$ as required. Hence, we have a Steiner tree for the MCST instance $(V', E', T', 2.\zeta)$ of weight at most $X$.
	
	For the other direction, let $H'$ be a Steiner tree of minimum internal degree at least $2.\zeta$ and $wt(H') \le X$. Define the sub-tree $H$ to be the graph obtained upon deleting all the $\eta^i_j$ vertices and the edges of $H'$ incident upon them.
	Even after these deletions, we can see that the resultant tree $H$ contains all the original terminals $T$ and is of weight at most $X$. We claim that the size of $H$ is at most $\zeta$. For contradiction, say its size were some $s \ge \zeta +1$. This means that the number of internal nodes of $H'$ is at least $s - |T|$. Combined with the fact that the degree of the internal nodes of $H'$ is at least $2.\zeta$, using Lemma \ref{lem:5}, we have that the number of leaves in $H'$ is at least $(2.\zeta - 2).(s- |T|)\ge (2.\zeta -1).(\zeta + 1 - |T|)$. But we know that the number of leaves is exactly the number of terminals in the new graph, i.e, $|T'| = |T| + |N| = |T| + (\zeta - |T|).2.\zeta$. Using these two bounds, we get
	
	$$|T| + (\zeta - |T|).2.\zeta \ge (2.\zeta -2).(\zeta + 1 - |T|)$$
	
	which upon simplifying yields $|T| \le 2$, which is false for the interesting/non-trivial instances of the problem. Hence, the size of $H$ is at most $\zeta$, meaning it is a valid SCST solution for $(V, E, T, \zeta)$ with weight at most $X$.
	
	The run-time of the reduction is polynomial in the input sizes as $|N|=2\zeta^2 \le 2|V|^2=poly(|V|)$ and defining the weighted adjacency matrix $E'$ and $T'$ only demands an overhead of linear  amount in $|E|$ and $|N|$.
	
\end{proof}

\section{Euclidean Steiner tree problem}
\label{sec:5}

Till now, we have seen the combinatorial or graph theoretic aspects of Steiner trees. However, there are wonderful insights to obtain from a more geometric definition of the Steiner problem. The combinatorial version may be seen a discrete precursor in understanding the Euclidean version. In fact, it will be clear that the Euclidean version can be approximated naturally by the discrete version by making the full vertices set as a grid in the plane and giving Euclidean distances as the edge weights. The terminals would simply be the vertices corresponding to the input points of the Euclidean version. We start by giving some definitions followed by a standard algorithm to compute the Euclidean Steiner tree.

\subsection{Definition}
\noindent\fbox{
	\parbox{13.5cm}{
		{\sc Euclidean Steiner Tree (EST)}  \\
		\textsf{\bfseries Input:} $N$ points in a Euclidean plane \\
		\textsf{\bfseries Output:} Minimum length of a (possibly sharp) curve passing through the input points.
	}
}

\begin{figure}[t!]
	\centering
	\begin{subfigure}[b]{1\textwidth}
		\includegraphics[height=2.8in]{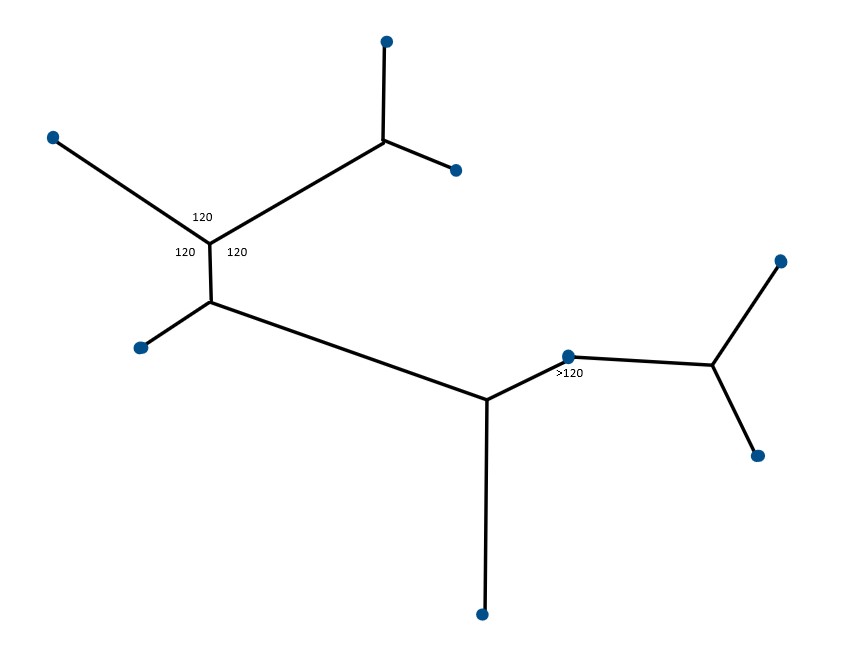}
		\caption{A solution to eight input points}
		\label{fig:4b}
	\end{subfigure}
	
	\begin{subfigure}[b]{1\textwidth}
		\includegraphics[height=2.8in]{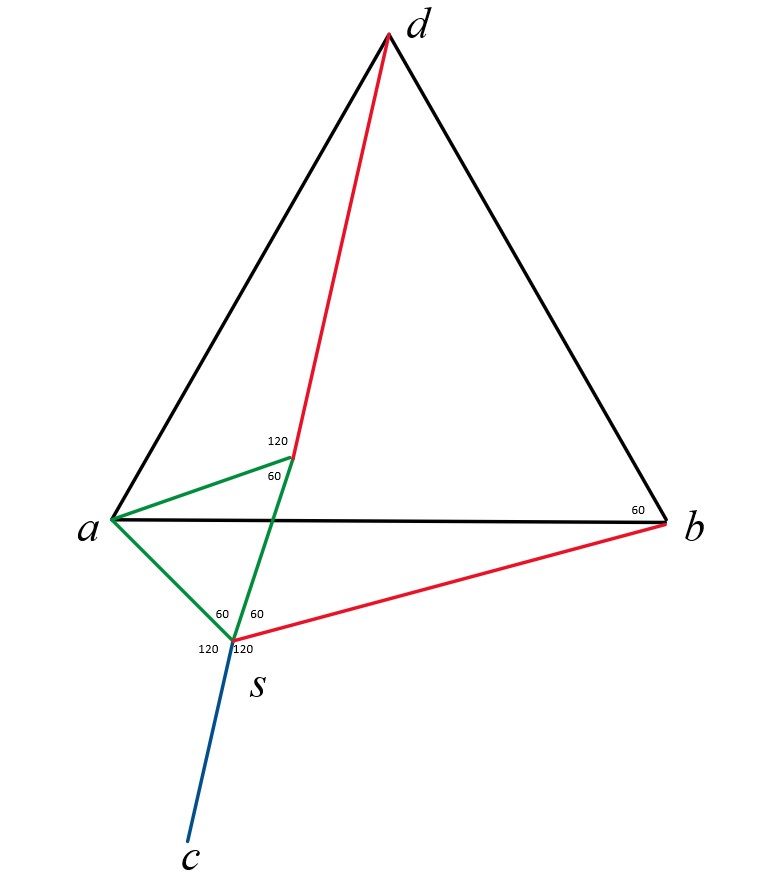}
		\caption{Remove $a,b$ and add $d$ to find the solution}
		\label{fig:4a}
	\end{subfigure}
	\caption{Euclidean Steiner tree}
\end{figure}

\vspace{2mm}

A critical concept to understand to solve the Euclidean version of the Steiner problem is the \textit{Fermat point} of a triangle. It is a point from which the total distance to the three vertices of the triangle is minimum. This point is inside the triangle if all the angles of the triangle are less than $120^\circ$. Otherwise the Fermat point is the vertex that is greater than $120^\circ$. An important property in the former case is that the three edges subtend an angle exactly equal to $120^\circ$ at the Fermat point (see \cite{fermat}).

Note that the Euclidean definition is different from the discrete version of {\sc ST} as there can be extra auxiliary points (called Steiner points) in the minimum cost solution (If no auxiliary points were allowed, it degenerates to the minimum spanning tree problem). For $N=3$, the solution can be computed easily: the only Steiner point is the Fermat point of the triangle. To argue the correctness, we appeal to the definition of Fermat point and the fact that a straight line segment is the shortest curve connecting any two given points. This idea can be used to arrive at the fact that the minimum Steiner tree can have at most $N-2$ ~Steiner points and each `Steiner junction' is $120^\circ$-symmetric (see Figure \ref{fig:4b}). We are now ready for a high level overview of one of the first exact algorithms for this problem \cite{Melzak1961OnTP}. There have been many other polynomial time approximation schemes and heuristics for this problem \cite{eucliedeanoverview, euclideanhistory, Soothill2010TheES}. 

\subsection{Melzak's algorithm - An overview}

A \textit{topology} of a  Steiner tree $T$ is the description of the structure of $T$. This includes specifying the whole adjacency list of $T$, without the distances between the points or the co-ordinates of the Steiner points of $T$. The idea is to go over all possible topologies and search for the optimal Steiner tree.

 Suppose we are given a Steiner tree's topology. Take two vertices, say $a$ and $b$ adjacent to some Steiner point $s$ (see Figure \ref{fig:4a}). We may assume that such an $s$ exists as otherwise, we can directly compute the minimum spanning tree for the given points and output it. Had we known the exact co-ordinates of $s$, we could simply remove the points $a$ and $b$ and introduce $s$ to the set of points, making the total number points $N-1$. Once we find a Steiner tree corresponding to these $N-1$ points, we can add the length $\overline{as} + \overline{bs}$ to it to get the final answer. Applying this recursively, we will get a polynomial time algorithm. The problem, however is that we only know the topology of the solution and we have no idea about the exact co-ordinates of the Steiner points. However, we can still do something very similar. The main idea is that in the optimal solution, the value $\overline{as} + \overline{bs} + \overline{cs}$ would be equal to $\overline{cd}$, where $c$ is the third (auxiliary or original) point that is involved at the Steiner junction of $s$ and $d$ is a point such that the triangle $\triangle abd$ is equilateral. This is evident once we notice some equi-length segments, which are colour coded in Figure \ref{fig:4a}. By this observation, we can remove the two points $a$ and $b$, and add the point $d$ and solve the problem for these $N-1$ points. But there would be two potential choices for $d$, one on each side of the line $ab$. Thus, we will have to solve two sub-problems each of size $N-1$ in the recursion, resulting in a recursive algorithm with running time $2^N.poly(N)$ to find the output for a given topology. Directly using the calculation for the number of possible topologies, the overall running time of Melzak's algorithm turns out to be $2^N.poly(N).(2N-4)!/[2^{N-2}(N-2)!] = exp(O(N))$.

\section{Conclusion: A note about the reductions}
\label{sec:concl}
The reductions in Sections \ref{sec:3} and \ref{sec:4}, although were dealt as polynomial time reductions, they are much stronger. The constrained parameter of the harder problem only depends on the constraint parameter of the easier problem and the relationship is quite simple ($D+1 $ and $2\zeta $). This may termed as a parameterized reduction when parameterized by the constrained parameter (diameter/min-degree/size), but even that is a weak statement, since the runtime of the reduction is only polynomial, unlike FPT time in a parameterized reduction. Note that the reduction in Section \ref{sec:4} is a many-one reduction. One could try to improve the reduction in Section \ref{sec:3} to make it many-one too. Once could try to accomplish similar reductions DCST $\le$ SCST (so that DCST $\le$ MCST, by transitivity) or any of the other possible reductions, especially the reverse reductions. That would concretely establish that all these problems are indeed very closely related.

\vspace{3mm}

\noindent \textbf{Acknowledgments:} We thank Prof.~Narayanaswamy N.S.~for introducing us to this problem and for his valuable comments.



\bibliographystyle{plain}
\bibliography{ref}

\end{document}